\documentclass[conference]{IEEEtran}

\IEEEoverridecommandlockouts
\usepackage{cite}
\usepackage{amsmath,amssymb,amsfonts}
\usepackage{amsthm}
\usepackage{amsmath}
\usepackage{algorithm}
\usepackage{algorithmic}
\newtheorem{lemma}{Lemma}

\newtheorem{corollary}{Corollary}
\usepackage{graphicx}
\usepackage{array}
\usepackage{textcomp}
\usepackage{xcolor}
\def\BibTeX{{\rm B\kern-.05em{\sc i\kern-.025em b}\kern-.08em
    T\kern-.1667em\lower.7ex\hbox{E}\kern-.125emX}}
\begin{document}

\title{Throughput Maximization for Full Duplex Wireless Powered Communication Networks\\
\thanks{This work is supported by Scientific and Technological Research Council of Turkey Grant $\#$117E241.}
}
\author{\IEEEauthorblockN{Muhammad Shahid Iqbal}
\IEEEauthorblockA{\textit{Electrical and Electronics Engineering} \\
\textit{Koc University}\\
Istanbul, Turkey \\
miqbal16@ku.edu.tr}
\and
\IEEEauthorblockN{Yalcin Sadi}
\IEEEauthorblockA{\textit{Electrical and Electronics Engineering} \\
\textit{Kadir Has University}\\
Istanbul, Turkey \\
yalcin.sadi@khas.edu.tr }
\and
\IEEEauthorblockN{Sinem Coleri}
\IEEEauthorblockA{\textit{Electrical and Electronics Engineering} \\
\textit{Koc University}\\
Istanbul, Turkey \\
scoleri@ku.edu.tr}
}
\maketitle

\begin{abstract}
In this paper, we consider a full duplex wireless powered communication network where multiple users with RF energy harvesting capabilities communicate to a hybrid energy and information access point. An optimization framework is proposed with the objective of maximizing the sum throughput of the users subject to energy causality and maximum transmit power constraints considering a realistic energy harvesting model incorporating initial battery levels of the users. The joint optimization of power control, time allocation and scheduling is mathematically formulated as a mixed integer non linear programming problem which is hard to solve for a global optimum. The optimal power and time allocation and scheduling decisions are investigated separately based on the optimality analysis on the optimization variables. Optimal power and time allocation problem is proven to be convex for a given transmission order. Based on the derived optimality conditions, we propose a fast polynomial-time complexity heuristic algorithm. We illustrate that the proposed algorithm performs very close-to-optimal while significantly outperforming an equal time allocation based scheduling scheme.
\end{abstract}


\section{Introduction} \label{sec:introduction}

Lifetime of a sensor network is generally battery dependent and requires replacement of or recharging the batteries once depleted. The replacement can be difficult (for large networks), harmful (inside chemicals) or infeasible (inside living beings). Whereas, recharging can take place via energy harvesting from environment \cite{solar_EH} or wireless energy transfer (WET). Solar, wind and vibration are the possible choices for energy harvesting from environment, however, their dependence on environmental conditions and requirement of extra large size equipment is the bottleneck. On the other hand, WET via inductive or magnetic coupling is constrained by short range, alignment issues and large size. Energy harvesting through radio frequency (RF) has long range, small form factor and better control on the power transfer along with mobility which makes it a suitable choice for such networks \cite{RF}. In this technology, electromagnetic waves in the spectrum of $300$Hz to $300$GHz are used to transfer the energy.
RF energy transfer have been studied as simultaneous wireless information and power transfer (SWIPT) and wireless powered communication networks (WPCN). In SWIPT, information is transferred in both direction and users use the received signal for both information decoding and energy harvesting by using time switching or power splitting \cite{SWIPT_RE2}. Whereas, in WPCN, a dedicated transmitter is deployed at a hybrid access point (HAP) for downlink energy transmission and users harvest this energy for their uplink data transmissions \cite{WPCN_HTT}.

The Sum Throughput Maximization (STM) problem is investigated in WPCNs for different models and protocols. \textit{Harvest-then-transmit} is the first protocol in which the frame length is divided into two non-overlapping time intervals, for energy and data transmission \cite{WPCN_HTT}. Variations of STM problems are investigated under different objectives such as weighted throughput \cite{WPCN_weightedThroughputmax}, common throughput \cite{WPCN_commonThroughputmax} and minimum throughput \cite{WPCN_minThroughputmax} maximization. All of these WPCN studies consider an half duplex model and all users get equal time to harvest energy; therefore, no scheduling is required for STM. Besides, these studies lack the maximum transmit power constraint for the users. 
 
Recently, full duplex (FD) became an active research area and a major technique being proposed for 5G and beyond networks. In the context of WPCN, FD is investigated in \cite{FD_WPCN_HD} where the HAP is operating in FD mode while the users are operating in half duplex mode which is extended for FD HAP and users in \cite{FD_WPCN_FD_transceiver}. The FD-WPCN throughput is maximized for point to point links in \cite{FD_WPCN_Kang} for a predetermined order under perfect self interference cancellation. The HAP with multiple antennas is investigated with the objective to know the effect of beamforming weights in \cite{WPCN_beamforming}. Due to FD, a user can harvest energy during both the transmission of other users and its own transmission, making scheduling important which is missing in the aforementioned studies.
The authors in \cite{WPCN_scheduling_Kalpant, WPCN_scheduling_Jie} considered scheduling for WPCN in a restricted fashion such as authors in \cite{WPCN_scheduling_Kalpant} divided frame in $M$ equal length slots and $K$ users are scheduled in these fixed length slots which leads to underutilization of the resources. In \cite{WPCN_scheduling_Jie}, authors have used Hungarian algorithm to find the schedule for the network with a predefined sum throughput limit. The determination of the cost matrix used by the Hungarian algorithm is a major computational burden for such schedule dependent energy harvesting problems with exponential-time complexity; hence, the application of Hungarian algorithm is not an efficient choice. Moreover, authors have only considered the frame length constraint which makes their problem simplistic compared to the problem model discussed in this paper.

The goal of this paper is to investigate the optimal time allocation, power control, and scheduling with the objective of maximizing the sum throughput for a set of energy harvesting users subject to the maximum transmit power and the energy causality constraints considering initial battery levels and a realistic energy harvesting model for a WPCN, in an in-band full-duplex system. The main contributions of the paper are listed as follows:

\begin{itemize}
\item We propose a new WPCN optimization framework for an in-band full-duplex energy harvesting network, employing the maximum transmit power constraint and an energy causality constraint which considers the initial battery levels and the energy harvesting capability of the users.
\item We characterize Sum Throughput Maximization Problem (STMP) based on the proposed WPCN model. We mathematically formulate the problem as a mixed integer nonlinear programming (MINLP) problem which is non-convex and thus generally difficult to solve for a global optimum. Then, we propose a solution framework based on the decomposition of the optimal power and time allocation and the scheduling decisions based on optimality analysis on the variables of the optimization problem.
\item We characterize the optimal power and time allocation problem (PTAP). PTAP gives the optimal solution of STMP for a given transmission order of the users. We further prove the convexity of PTAP which guarantees a polynomial-time complexity optimal algorithm.
\item We analyze the optimality conditions on the optimization variables of STMP. Based on this analysis, we propose a polynomial-time heuristic algorithm. We illustrate that the proposed algorithm performs very close to optimal for different simulation scenarios.
\end{itemize}

The rest of this paper is organized as follows. In Section \ref{sec:system_model}, we describe the WPCN model and assumptions used throughout the paper. In Section \ref{sec:problem}, we present the mathematical formulation of the sum throughput maximization problem and discuss its complexity. In Section \ref{sec:power}, we present the optimal power and time allocation problem and prove its convexity. In Section \ref{sec:algorithm}, we analyze the optimality conditions of the problem and propose a polynomial-time heuristic algorithm. In Section \ref{sec:performance}, we evaluate the performance of the proposed algorithm. Section \ref{sec:conclusion} presents the concluding remarks.

\section{System Model and Assumptions} \label{sec:system_model}

The system model and assumptions are described as follows:

\begin{figure}[t]
\centering
\includegraphics[width= 0.8 \linewidth]{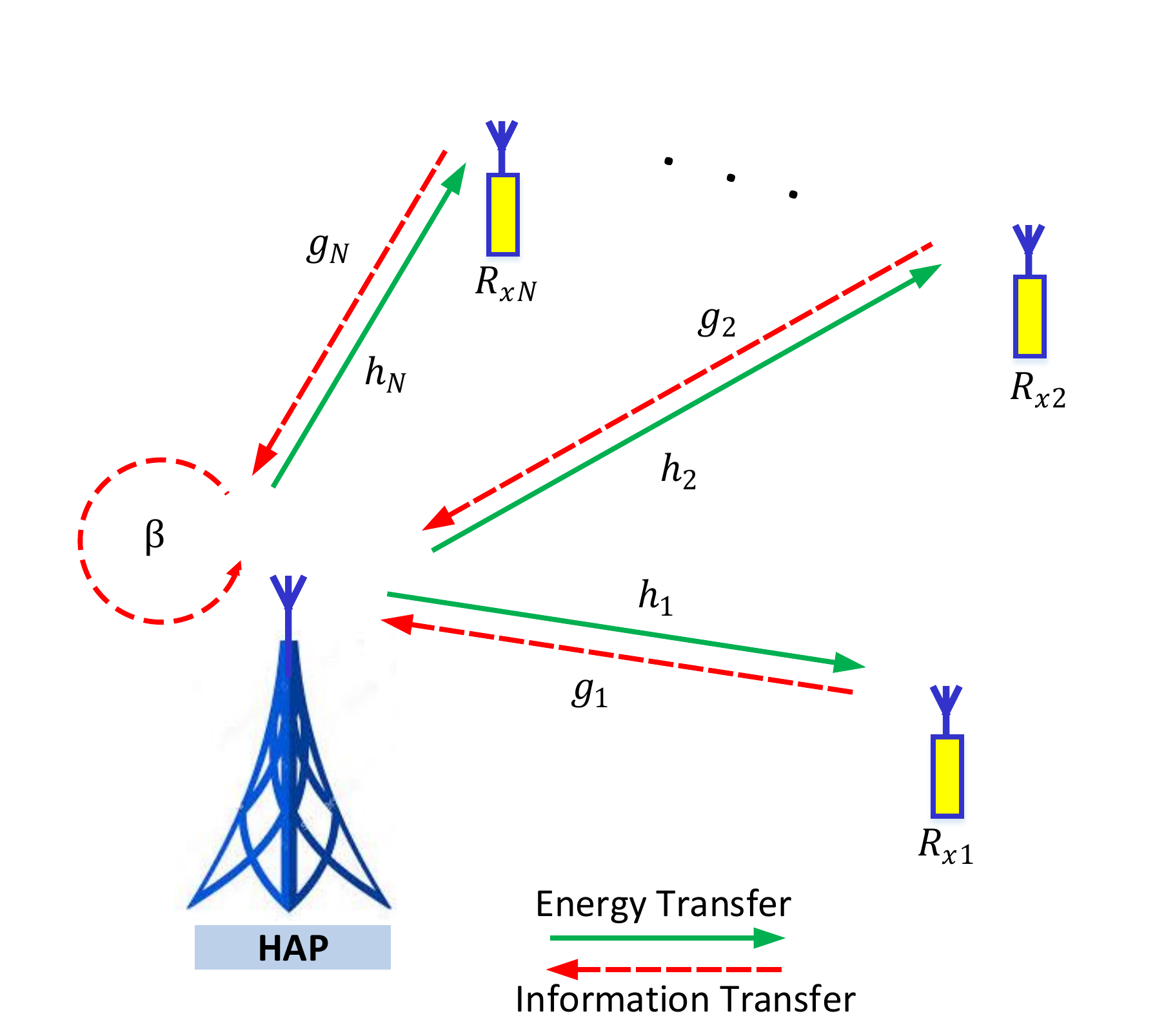}
\caption{System Model for Wireless Powered Communication Network} \label{fig:WPCN}
\end{figure}

\begin{enumerate}
\item The WPCN archictecture, as depicted in Fig. \ref{fig:WPCN}, consists of a HAP and N users; i.e., sensors and machine type communications devices. The HAP and the users are equipped with one full-duplex antenna for simultaneous wireless energy transfer and data transmission on downlink and uplink channels, respectively. The channel gains for the uplink and downlink channels are assumed to be different. The uplink channel gain from user $i$ to the HAP and the downlink channel gain from the HAP to user $i$ are denoted by $g_i$ and $h_i$, respectively. 
\item The HAP has a stable energy supply and continuously radiates wireless energy with a constant power $P_h$ whereas the users do not have any external energy supply. Each user $i$ harvests the radiated energy from the HAP and stores in a rechargeable battery of capacity $B_{max}$ which is assumed to be large enough so that no overflow will occur. Each user has an initial energy $B_i$ stored in its battery at the beginning of the scheduling frame which includes the harvested and unused energy in the previous scheduling frames.
\item The energy harvesting rate of user $i$ from the HAP, denoted by $C_i$, is characterized as follows:
\begin{equation} \label{eq:en_harv_rate}
C_i=\eta_i h_i P_h
\end{equation}
where $\eta_i$ is the antenna efficiency of user $i$.
\item We consider Time Division Multiple Access (TDMA) as medium access control for the uplink data transmissions from the users to the HAP. The time is partitioned into scheduling frames which are further divided into variable-length time slots each of which is allocated to a particular user. 
\item We use continuous rate transmission model, in which Shannon capacity formulation for an AWGN channel is used in the calculation of transmission rate $r_i$ of user $i$ as
 \begin{equation} \label{transmission_rate}
 r_i = W \log_{2}(1+k_iP_i),
 \end{equation}
where $P_i$ is the transmission power of user $i$, $W$ is the channel bandwidth, and $k_i$ is defined as $g_i/(N_oW+\beta P_h)$. The term $\beta P_h$ is the self interference at the HAP and $N_o$ is the noise power density.
\item We use continuous power model in which the transmission power $P_i$ of user $i$ can take any value below a maximum level $P_{max}$, which is imposed to limit the interference to nearby systems.
\end{enumerate}

\section{Sum Throughput Maximization Problem} \label{sec:problem}

In this section, we introduce the sum throughput maximization problem, denoted by STMP. The joint optimization of the time allocation, power control and scheduling with the objective of maximizing the sum throughput of the users is formulated as follows:\\

STMP:
\begin{subequations} \label{opt_problem}
\begin{align}
&\textit{maximize}
& & \sum_{i=1}^{N} W \tau_i \log_2(1+k_iP_i) \label{obj}\\
& \textit{subject to}
&& B_i+C_i\tau_0+C_i \sum_{j=1}^{N}a_{ji}\tau_j+C_i \tau_i -P_i\tau_i \geq 0,  \label{energyharvesting}\\
&&& a_{ij}+a_{ji}=1, \label{ordering}\\
&&& P_i\leq P_{max}, \label{pmax} \\
&&& \sum_{i=0}^{N}\tau_i\leq 1, \label{framelength} \\
& \textit{variables}
& & P_i \geq 0, \hspace*{0.1cm} \tau_i\geq 0, \hspace*{0.1cm} a_{ij} \in \{0,1\}.\label{vars}
\end{align}
\end{subequations}

The variables of the problem are $P_i$, the transmit power of user $i$; $\tau_i$, the transmission time of user $i$, and $a_{ij}$, binary variable that takes value $1$ if user $i$ is scheduled before user $j$ and $0$ otherwise. In addition, $\tau_0$ denotes an initial waiting time in which all users harvest energy without transmitting data.

The objective of the problem is to maximize the sum throughput of the users as given by Eq. (\ref{obj}). Eq. (\ref{energyharvesting}) is the energy causality constraint; i.e., the energy consumed by a user during transmission cannot exceed the total amount of available energy for that user including its initial battery level and the harvested energy until and during its transmission. Eqs. (\ref{ordering}) and (\ref{pmax}) represent the scheduling and maximum transmit power constraints, respectively. Eq. (\ref{framelength}) indicates that normalized schedule length of $1$ is allocated to the users.

In the following we illustrate that in the optimal solution, an initial waiting time in which all users only harvest energy is not needed.

\begin{lemma} \label{lemma:zerotau}
In the optimal solution of STMP, $\tau_0=0$.
\end{lemma}

\begin{proof}
Suppose that in the optimal solution, we have $\tau_0^*>0$ and without loss of generality we have the following transmission schedule: $\tau_{1}^{*},\tau_{2}^{*},...,\tau_{N}^{*}$. Denote the sum throughput of the optimal solution by $R^{*}=\sum_{i=1}^N R_i^{*}$. Then, consider that we have an alternative schedule keeping the same transmission order in which $\tau_0^{**}=0$, $\tau_1^{**}=\tau_1^{*}+\tau_0^{*}$, and $\tau_j^{**}=\tau_j^{*}$ for all $j\in \lbrace 2,..., N\rbrace$. Denote the sum throughput of this schedule by $R^{**}=\sum_{i=1}^N R_i^{**}$. Clearly, $R_j^{**}=R_j^{*}$ since $\tau_j^{**}=\tau_j^{*}$ for all $j\in \lbrace 2,..., N\rbrace$. However, $R_1^{**}= \tau_1^{**} \log(1+k_1 E_1/\tau_1^{**})>\tau_1^{*} \log(1+k_1 E_1/\tau_1^{*})=R_1^{*}$ since the throughput of a user monotonically increases by decreasing transmit power for a constant energy consumption. Then, $R^{**}>R^{*}$. This is a contradiction.
\end{proof}


STMP formulation is a MINLP problem thus difficult to solve for the global optimum \cite{opt_book}. In other words, finding a global optimum requires solutions with exponential time complexity which are intractable. Therefore, in the following, we follow a solution strategy based on the decomposition of the power and time allocation and the scheduling decisions and on the optimality analysis on the variables of the optimization problem.

\section{Optimal Power and Time Allocation} \label{sec:power}

In this section, we investigate optimal power control and time allocation problem, denoted by PTAP, to maximize the sum throughput for a given scheduling order of a set of users; i.e. $a_{ij}$s are given in STMP. Without loss of generality, we assume that user $i$ transmits during the time slot $i$.\\

PTAP:
\begin{subequations} \label{opt_problem2}
\begin{align}
&\textit{maximize}
& & \sum_{i=1}^{N} W \tau_i \log_2(1+k_i P_i) \label{obj2}\\
& \textit{subject to}
&& B_i+C_i \sum_{j=1}^{i} \tau_j  -P_i\tau_i \geq 0,  \label{energyharvesting2}\\
&&& P_i\leq P_{max}, \label{pmax2} \\
&&& \sum_{i=0}^{N}\tau_i\leq 1, \label{framelength2} \\
& \textit{variables}
& & P_i \geq 0, \hspace*{0.1cm} \tau_i\geq 0.\label{vars2}
\end{align}
\end{subequations}

\begin{lemma} \label{lemma:p2con}
PTAP is a convex optimization problem.
\end{lemma}

\begin{proof}
The throughput function for user $i$,  $W \tau_i \log_2(1+k_i P_i)$ is a concave function of $\tau_i$ and $P_i$ since its Hessian is negative semi-definite. Since a nonnegative sum of concave functions is concave, the objective function of PTAP is concave. The constraints (\ref{pmax2}) and (\ref{framelength2}) are affine and the left hand side of constraint (\ref{energyharvesting2}) is concave due to negative semi-definiteness of its Hessian. Therefore, PTAP is a convex optimization problem. 
\end{proof}

Since PTAP is a convex optimization problem, it can be solved by common convex optimization techniques which have a computational complexity on the order of $\mathcal{O}(N^3)$.

Once the convexity of PTAP is proven, a straigthforward solution to determine the optimal schedule, i.e., to optimally solve STMP, would be an exhaustive search algorithm which enumerates all possible transmission orders for the set of users and then determine the one maximizing the sum throughput by solving each of them using convex programming. However, such an optimal algorithm has an exponential complexity of $\mathcal{O}(N! N^3)$ (due to $N!$ possible transmission orders) which makes it computationally intractable. Hence, in the following, we present a polynomial time complexity heuristic algorithm by investigating the characteristics of an optimal solution.

\section{Scheduling Algorithm} \label{sec:algorithm}
In this section, we propose a polynomial-time algorithm based on the optimality conditions of the STMP derived in the following.

The following lemma proves that sum throughput maximization requires utilization of the scheduling frame fully.

\begin{lemma} \label{lemma:length}
In the optimal solution of STMP, the constraint (\ref{framelength}) must be staisfied with equality; i.e., $\sum_{i=1}^{N}\tau_i=1$.
\end{lemma}

\begin{proof}
Suppose that in the optimal solution, we have $\sum_{i=1}^{N}\tau_i^*<1$. Let $\tau_{\Delta}=1-\sum_{i=1}^{N}\tau_i^*$. Suppose that the optimal schedule is updated by introducing an unallocated time slot $\tau_0=\tau_{\Delta}$ such that $\tau_0+\sum_{i=1}^{N}\tau_i^*=1$ without changing the transmission power of the users. Then, by Lemma \ref{lemma:zerotau}, the schedule can be updated to yield a greater sum throughput. This is a contradiction.
\end{proof}

Next, we present the optimality condition on the power and time allocation for a user.

\begin{lemma} \label{lemma:powerallocation}
In the optimal solution of STMP, either constraint ($\ref{energyharvesting}$) or constraint $\ref{pmax}$ must be satisfied with equality; i.e., a user either transmits with maximum transmit power $P_{max}$ or consumes all its energy available until the completion of its transmission.
\end{lemma}

\begin{proof}
Suppose that $\tau^*=[\tau_1^*,\tau_2^*,...,\tau_N^*]$ and $P^*=[P_1^*,P_2^*,...,P_N^*]$ are optimal transmission time and power allocation for a set of users. Further suppose that for a user $k$, $P_k^*<P_{max}$ and $P_k^* \tau_k^* < B_k+C_k \sum_{j=1}^{k} \tau_k^*$. Then, while keeping $\tau_k^*$ constant, $P_k^*$ can be increased until it first becomes equal to either $P_{max}$ or $\frac{B_k+C_k \sum_{j=1}^{k}\tau_j^* }{\tau_k^*}$. Then, throughput of user $k$ increases since it is a monotonically increasing function of transmission power for a constant transmission time. This is a contradiction.
\end{proof}

In the following lemma, we illustrate an optimality condition on scheduling suggesting a prioritization among users based on their maximum transmission rates.

\begin{lemma} \label{lemma:rateordering}
Let $r_i^{max}$ be the transmission rate of user $i$ corresponding to maximum transmit power $P_i=P_{max}$; i.e., $r_i^{max}= W \log_{2}(1+k_i P_{max})$. Then, in the optimal solution of STMP, for any two users $i$ and $j$ such that $r_i^{max}<r_j^{max}$, if $\tau_i > 0$, then $\tau_j > 0$.
\end{lemma}

\begin{proof}
Suppose that $\tau^*=[\tau_1^*,\tau_2^*,...,\tau_N^*]$ and $P^*=[P_1^*,P_2^*,...,P_N^*]$ are optimal time and power allocation such that $\tau_i^*=0$ and $\tau_j^*>0$ for some $i$ and $j$ such that $r_i^{max}>r_j^{max}$. Then, for any $\tau^{'} < \min \lbrace B_i/P_{max}, \tau_j^* \rbrace$, $\tau^*$ and $P^*$ can be updated as $\tau_i^*=\tau^{'}$, $\tau_j^*=\tau_j^*-\tau^{'}$, and $P_i^*=P_{max}$. Then, throughput is increased by $R^{'}=\tau^{'} (r_i^{max}-r_j^*)>\tau^{'} (r_i^{max}-r_j^{max})>0$. This is a contradiction.
\end{proof}

Lemma \ref{lemma:rateordering} suggests that users with higher maximum rates that can contribute to the sum throughput more should be prioritized by a scheduling algorithm. Note that the optimal schedule does not necessarily contain a time slot for each and every user in the network as long as the maximum sum throughput is achieved using a subset of users. However, for instance, the user with maximum transmission rate should be given nonzero time slot length as stated by the following corollary of Lemma \ref{lemma:rateordering}.

\begin{corollary} \label{cor:1}
Let user $k$ has $r_k^{max}=\max_i r_i^{max}$. Then, there exists an optimal solution in which $\tau_k>0$.
\end{corollary}

Depending on the energy available for users; i.e., initial battery levels and energy harvesting capabilities, some high rate users may utilize the entire scheduling frame. For instance, if the maximum rate user can afford to transmit with $P_{max}$ for the entire scheduling frame without violating the energy causality constraint, the optimal schedule will only contain this particular user as indicated by the following corollary of Lemma \ref{lemma:rateordering}.

\begin{corollary} \label{cor:2}
Let user $k$ has $r_k^{max}=\max_i r_i^{max}$. Then, if $B_i/(P_{max}-C_i)\geq 1$, there exists an optimal solution in which $\tau_k=1$; i.e. $k$ is allocated to the entire scheduling frame.
\end{corollary}

Next, based on the foregoing analysis, we propose Maximum-Rate First Scheduling Algorithm (MFSA), given in Algorithm \ref{algo_MFSA}. Suppose a set $\cal{F}$ of users are sorted in decreasing order of maximum transmission rates; i.e. $r_1^{max}\geq r_2^{max} \geq ... \geq r_N^{max}$. As suggested by Lemma \ref{lemma:rateordering} and its corollaries, MFSA starts allocation from the maximum rate user and iteratively determines the time and power allocation of each user in decreasing order of maximum rates. Note that MFSA allocates the users starting from the end of the schedule to allow higher rate users harvest more energy. As indicated by Corollary \ref{cor:2}, if maximum rate user can afford to transmit with $P_{max}$, it will be allocated to the entire scheduling frame and no other user will be given any time slot (Lines $7-10$). Otherwise, it determines if next user should be given a nonzero time slot. This is done by pairwise evaluation between two users, users $i$ and $i-1$. For any time duration $\tau$ to be allocated between users $i$ and $i-1$, there are $3$ different optimality cases due to convexity of throughput as a function of time duration for constant energy consumption (we omit detailed analysis due to space limitations). Note that the user with higher rate, user $i$, should be allocated at least for a duration $\tau_i^{min}$ which is the maximum time duration it can afford $P_{max}$. In case $1$, user $i$ is allocated with $P_{max}$ in a time slot $\tau_i^{min}$ (Line $12$) and the rest of the available time duration $t_a$ is allocated to user $i-1$ with maximum feasible transmit power (Line $13$). In case $2$, user $i-1$ is allocated with $P_{max}$ in maximum feasible time slot length without violating the energy causality constraint and considering that user $i$ should be allocated at least for a duration $\tau_i^{min}$ (Line $15$). The rest of the available time duration $t_a$ is allocated to user $i$ with maximum feasible transmit power (Line $16$). In case $3$, user $i$ is allocated to the entire available time duration $t_a$ with maximum feasible transmit power and user $i-1$ is given no time slot (Lines $18-19$). Then MFSA evaluates which case yields the maximum sum throughput for users $i$ and $i-1$ (Lines $20-21$). If case $3$ is optimal, then it means that user $i-1$ cannot be given any time slot which indicates that lower rate users also cannot be allocated due to Lemma \ref{lemma:rateordering}. Hence, MFSA terminates (Lines $22-25$). If case $2$ is optimal, it means user $i-1$ is allocated with $P_{max}$ which hinders the possibility of allocating next lower rate users with the same pairwise evaluation due to Corollary \ref{cor:2} and MFSA terminates (Lines $22-25$). If case $1$ is optimal, then user $i$ is allocated with $P_{max}$ to maximum feasible time duration (Line $27$), the available time duration for the rest of the users is updated (Line $28$), and MFSA continues with next two highest rate users, mainly users $i-1$ and $i-2$. The computational complexity of MFSA is $\mathcal{O}(N)$.

\begin{algorithm} 
\caption{Maximum-Rate First Scheduling Algorithm}  \label{algo_MFSA}
\begin{algorithmic}[1] 
\STATE \textbf{Input:} set of users $\cal{F}$
\STATE \textbf{Output:} $\tau^*=[\tau_1^*,\tau_2^*,...,\tau_N^*]$, $P^*=[P_1^*,P_2^*,...,P_N^*]$
\STATE $\tau^*=[0,0,...,0]$, $P^*=[0,0,...,0]$,
\STATE $t_a=1$, 
\FOR {$i=1:{|\cal{F}|}-1$}
\STATE $\tau_{i}^{min}=(B_i+C_i t_a)/P_{max}$,
\IF{$\tau_{i}^{min}\geq t_a$}
\STATE $\tau_i^*=t_a$, $P_i^*=P_{max}$,
\STATE break,
\ENDIF
\STATE \textbf{Case 1}:
\STATE $\tau_{i,1}=\tau_{i}^{min}$, $P_{i,1}=P_{max}$,
\STATE $\tau_{i-1,1}=t_a-\tau_{i,1}$,\\ $P_{i-1,1}=\min\lbrace(B_{i-1}+C_{i-1} \tau_{i-1,1})/\tau_{i-1,1}, P_{max}\rbrace$,
\STATE \textbf{Case 2}:
\STATE $\tau_{i-1,2} =\min \lbrace B_{i-1}/(P_{max}-C_{i-1}),  (t_a-\tau_{i}^{min}) \rbrace$, \\$P_{i-1,2}=P_{max}$,
\STATE $\tau_{i,2}=t_a-\tau_{i-1,2}$, $P_{i,2}=(B_i+C_i t_a)/\tau_{i,2}$,
\STATE \textbf{Case 3}:
\STATE $\tau_{i,3}=t_a$, $P_{i,3}=(B_i+C_i t_a)/\tau_{i,3}$,
\STATE $\tau_{i-1,3}=0$, $P_{i-1,3}=0$,
\STATE determine $R_1$, $R_2$, and $R_3$,
\STATE $k =$ arg$\max\lbrace R_1, R_2, R_3\rbrace$,
\IF {$k = 2$ or $k = 3$}
\STATE $\tau_i^*=\tau_{i,k}$, $\tau_{i-1}^*=\tau_{i-1,k}$,
\STATE $P_i^*=P_{i,k}$, $P_{i-1}^*=P_{i-1,k}$,
\STATE break,
\ELSE
\STATE $\tau_i^*=\tau_{i,k}$, $P_i^*=P_{i,k}$,
\STATE $t_a=t_a-\tau_i^*$,
\IF{$i={|\cal{F}|}-1$}
\STATE $\tau_{i-1}^*=\tau_{i-1,k}$, $P_{i-1}^*=P_{i-1,k}$,
\ENDIF
\ENDIF
\ENDFOR
\end{algorithmic}
\end{algorithm}   

\section{Performance Evaluation} \label{sec:performance}

The goal of this section is to evaluate the performance of the proposed algorithm. The proposed scheduling algorithm MFSA is compared to optimal solution, denoted by OPT and equal time allocation based scheduling scheme, denoted by ETA. OPT is obtained by enumerating all possible transmission orders and determining the one with maximum sum throughput by solving PTAP problem for each. ETA allocates equal time to each user; i.e., $\tau_i = 1/N$, and allocates maximum feasible transmit power for each user such that either maximum transmit power or energy causality constraint is satisfied with equality as stated by Lemma \ref{lemma:powerallocation}.

Simulation results are obtained by averaging over $200$ independent random network realizations. The channel gains considering large-scale statistics are determined using the path loss model given by 
\begin{equation}
PL(d)=PL(d_0)+10\alpha log_{10}\bigg(\frac{d}{d_0}\bigg)+\emph{Z}
\end{equation}
where $PL(d)$ is the path loss at distance $d$ in $dB$, $d_0$ is the reference distance, $\alpha$ is the path loss exponent, and $Z$ is a zero mean Gaussian random variable with standard deviation $\sigma$. The Rayleigh fading has been used to model small-scale fading with scale parameter  $\Omega_i$ set to mean power level obtained from the large-scale path loss model. The parameters used in the simulations are $\eta_i=1$ $\forall i$, $W= 1$ MHz, $d_0=1$ m, $PL(d_0)=30$ dB, $\alpha=2.76$, and $\sigma=4$.

\begin{figure}[t]
\centering
\includegraphics[width= 0.8 \linewidth]{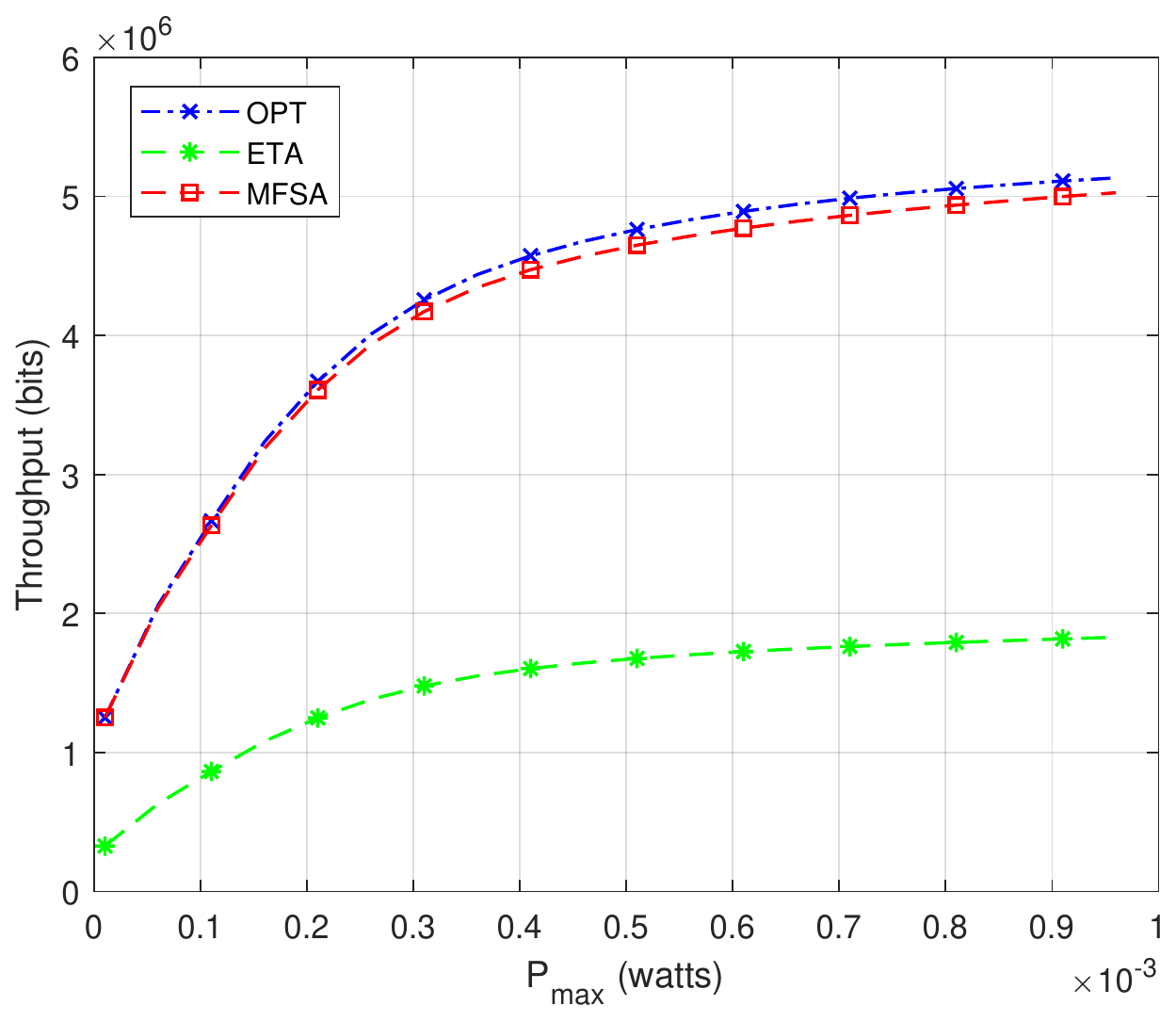}
\caption{Throuhput vs. maximum transmit power $P_{max}$} \label{fig:pmax}
\end{figure}

Fig. \ref{fig:pmax} illustrates the sum throughput for varying $P_{max}$ values in a network of $6$ users. As the figure clearly depicts, MFSA outperforms ETA significantly while achieving close-to-optimal performance. The superiority of MFSA over ETA is due to fact that MFSA favors the users with higher rate performance compared to low rate users while allocating transmission times to maximize the sum throughput while ETA does not consider the contribution of each user on the sum throughput and allocates time slots equally which degrades the performance dramatically. Small $P_{max}$ values favor high rate users to be given larger time slots compared to low rate users. For instance, if the highest rate user can afford to transmit with $P_{max}$ for the entire schedule, it will be scheduled accordingly by MFSA. Hence, the performance of MFSA is nearly optimal for smaller $P_{max}$ values. On the other hand, as $P_{max}$ increases, the expected transmission time for small rate users will increase, increasing the variation of the performance of MFSA from the optimal solution. However, MFSA still achieves robustness against this variance and resembles optimal solution for large $P_{max}$ values.

\begin{figure}[t]
\centering
\includegraphics[width= 0.8 \linewidth]{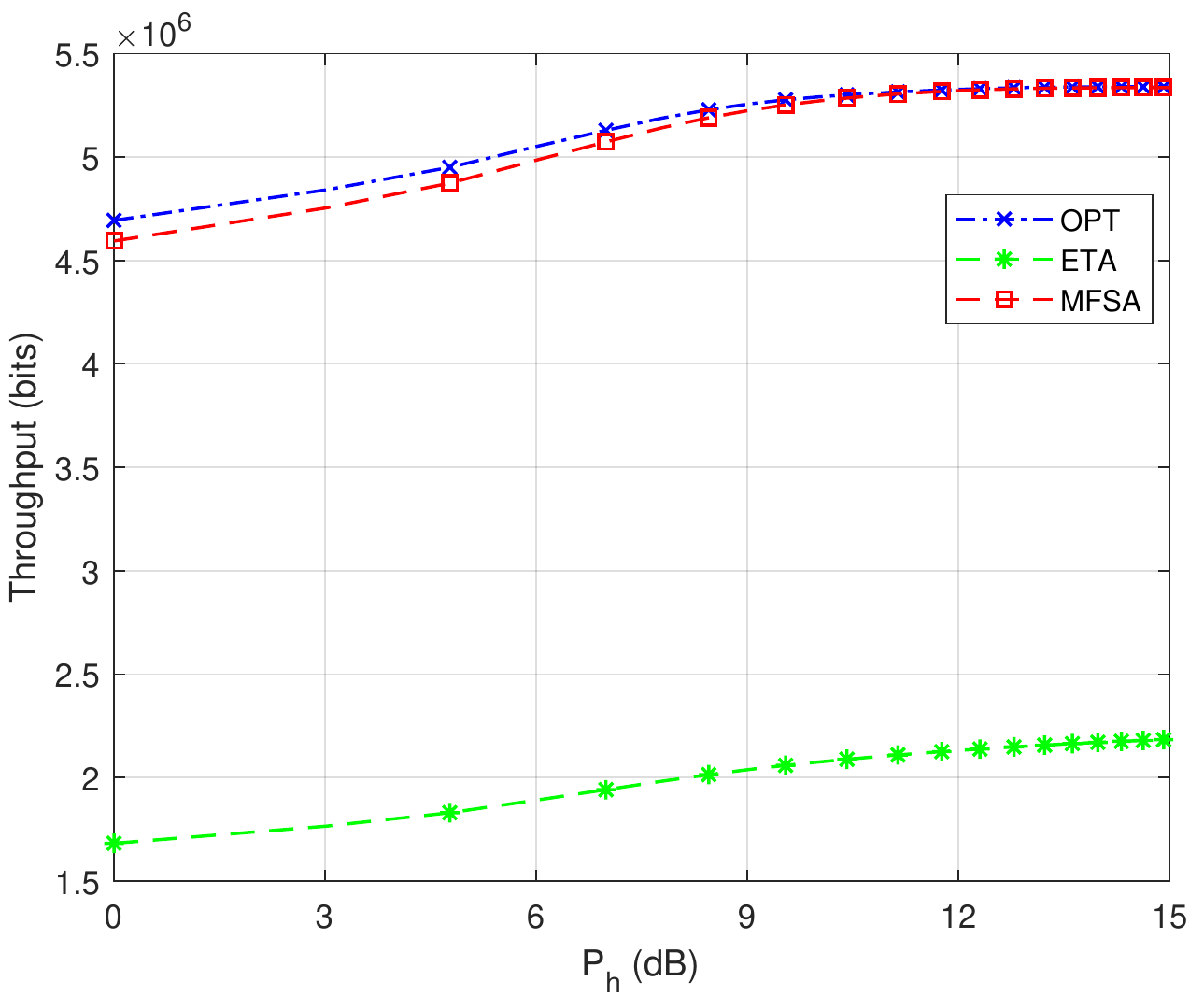}
\caption{Throuhput vs. HAP transmit power $P_{h}$} \label{fig:ph}
\end{figure}


Fig. \ref{fig:ph} illustrates the sum throughput for varying $P_{h}$ values in a network of $6$ users. High HAP power indicates that users can harvest more energy and achieve higher transmit power and rate until they are constrained by maximum transmit power constraint. Therefore, for higher HAP power values, users are expected to transmit with $P_{max}$ with higher probability and hence MFSA converges to the optimality. However, MFSA performs very close to optimal for a wide range of $P_h$ values. Moreover, MFSA significantly outperforms ETA. 

\begin{figure}[t]
\centering
\includegraphics[width= 0.8 \linewidth]{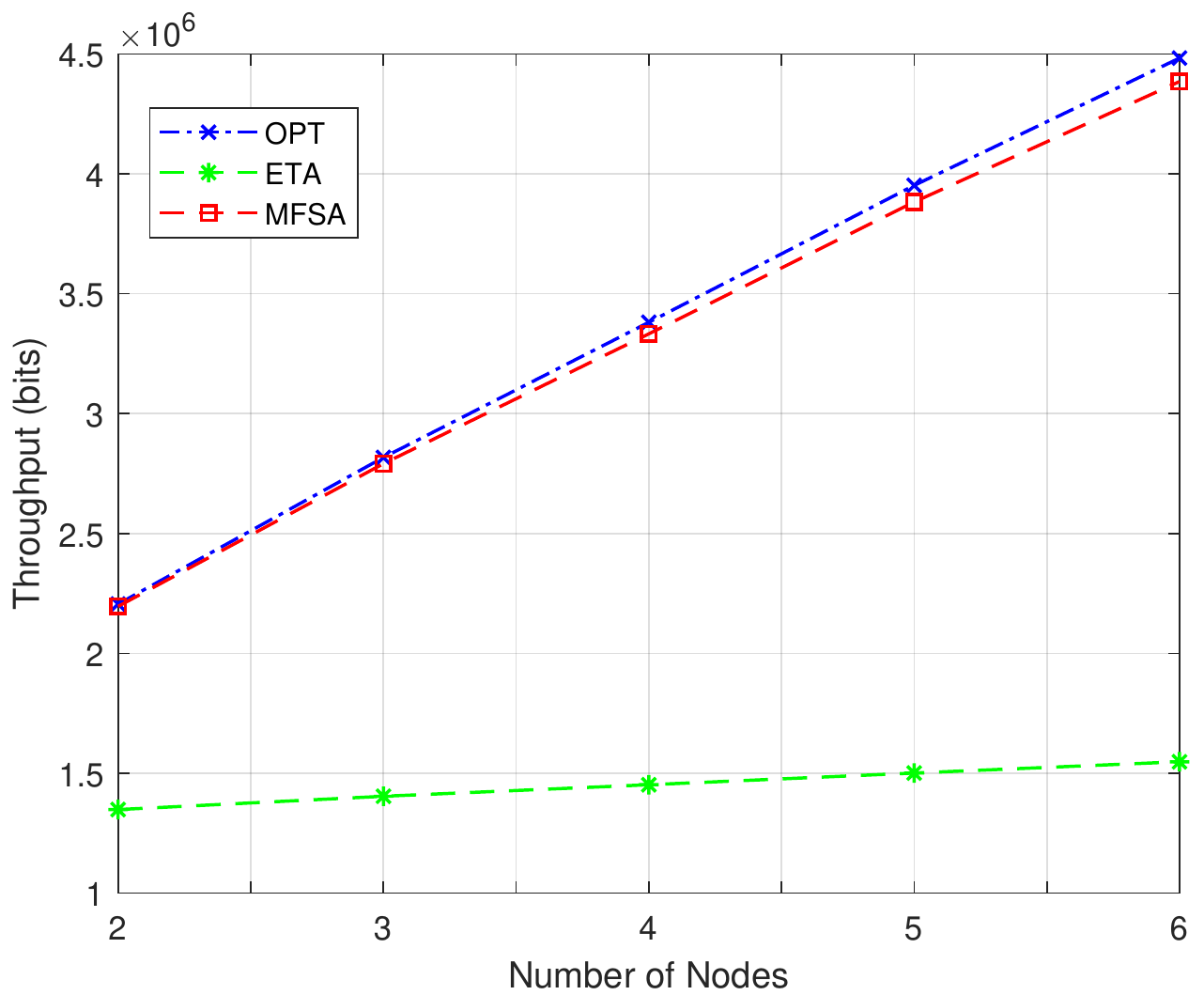}
\caption{Throuhput vs. number of users ($N$)} \label{fig:num}
\end{figure}

In Fig. \ref{fig:num}, we investigate the effect of the network size on the performance of the proposed algorithm. Throughput increase yielded by MFSA with the addition of each user is larger compared to the ETA algorithm showing the significant effect of variable size time allocation performed by MFSA on the performance. This suggests that ETA suffers from throughput performance of low rate users equally allocated with high rate users while MFSA can handle this problem by efficient allocation of time and power to the users based on their throughput performance. MFSA also shows robustness against increasing network size and performs very close to optimal.

\section{Conclusions} \label{sec:conclusion}

In this paper, we have investigated sum throughput maximization problem (SMTP) for a realistically modelled in-band full-duplex WPCN system in a new optimization framework employing the maximum transmit power constraint and an energy causality constraint which considers the initial battery levels and the energy storage capability of the users. We have mathematically formulated SMTP as an MINLP problem which is generally difficult to solve for a global optimum. To solve this intractability problem, we have provided a solution strategy in which the power control and time allocation, and the scheduling decisions are decomposed. We have also formulated the optimal power and time allocation problem for a given scheduling order of users and showed its convexity. Then, analyzing the optimality conditions on the power and time allocation and scheduling, we have proposed a polynomial-time complexity heuristic algorithm. Through extensive simulations, we have illustrated that the proposed algorithm performs very close to optimal while outperforming an equal time allocation based scheduling scheme significantly for various numerical scenarios. 

\bibliography{shahid_bib}
\bibliographystyle{ieeetr}
\end{document}